\documentclass[letterpaper, 10 pt, conference]{ieeeconf}  

\IEEEoverridecommandlockouts                              
\usepackage{cite}
\usepackage{amsmath,amssymb,amsfonts}
\usepackage{algorithmic}
\usepackage{graphicx}
\usepackage{algorithm,algorithmic}
\usepackage{hyperref}
\usepackage{cleveref}
\usepackage{textcomp}
\usepackage{float}
\usepackage{mathtools}
\usepackage{bm}

\usepackage{mathtools}
\usepackage{xurl}
\usepackage{mathrsfs}
\usepackage{xcolor}
\usepackage{float}
\usepackage{afterpage}

\newtheorem{definition}{Definition}
\newtheorem{theorem}{Theorem}

\newtheorem{lemma}{Lemma}

\newtheorem{assumption}{Assumption}

\newlength{\leftstackrelawd}
\newlength{\leftstackrelbwd}
\def\leftstackrel#1#2{\settowidth{\leftstackrelawd}%
{${{}^{#1}}$}\settowidth{\leftstackrelbwd}{$#2$}%
\addtolength{\leftstackrelawd}{-\leftstackrelbwd}%
\leavevmode\ifthenelse{\lengthtest{\leftstackrelawd>0pt}}%
{\kern-.5\leftstackrelawd}{}\mathrel{\mathop{#2}\limits^{#1}}}

\title{\LARGE \bf
Delay-adaptive Control of Nonlinear Systems\\ with Approximate Neural Operator Predictors 
}

\author{Luke Bhan, Miroslav Krstic, Yuanyuan Shi
\thanks{All three authors are with the University of California, San Diego.}
\thanks{\{lbhan, mkrstic, yyshi\}@ucsd.edu.}
\thanks{Code and datasets available at \url{https://github.com/lukebhan/Neural-Operator-Delay-Adaptive-Predictor-Feedback}.}}

\begin{document}

\maketitle
\thispagestyle{empty}
\pagestyle{empty}

\begin{abstract}
In this work, we propose a rigorous method for implementing predictor feedback controllers in nonlinear systems with unknown and arbitrarily long actuator delays. To address the analytically intractable nature of the predictor, we approximate it using a learned neural operator mapping. This mapping is trained once, offline, and then deployed online, leveraging the fast inference capabilities of neural networks. We provide a theoretical stability analysis based on the universal approximation theorem of neural operators and the transport partial differential equation (PDE) representation of the delay. We then prove, via a Lyapunov-Krasovskii functional, semi-global practical convergence of the dynamical system dependent on the approximation error of the predictor and delay bounds. Finally, we validate our theoretical results using a biological activator/repressor system, demonstrating speedups of $15$ times compared to traditional numerical methods.
\end{abstract}

\section{Introduction} 
In this work, we extend the concept of neural operator approximate predictors to delay systems with constant but \emph{unknown} delays. The presence of unknown delays presents a significant challenge, as the approximation of the predictor introduces an additive error requiring analysis akin to robust adaptive control. Furthermore, unlike robust adaptive control, the adaptive parameter estimated is the actuator delay which directly impacts system stability. 
Despite these challenges, we achieve a result similar to \cite{bhan2024neuraloperatorspredictorfeedback}, ensuring the semi-global practical stability of the feedback system, dependent on the approximation error and an additional error term depending on the delay bounds. This represents the first result for implementing any type of approximate predictor when the delay is constant but unknown. 

\subsection{Predictor feedback designs and implementations} \label{subsec:introduction-predictor-feedback}
Predictor feedback approaches for compensating actuator delays in dynamical systems have been been studied for well-over 50 years \cite{smith1957closer, henson1994time, https://doi.org/10.1002/aic.690350914,  935057, 1272267,  1272269, doi:10.1137/040616383}. Furthermore, predictor feedback compensation has benefited a variety of applications including traffic \cite{9987680, 6248674}, aerospace vehicles \cite{8930010, 9468382}, and robotics \cite{9727201, BEKIARISLIBERIS20131576, 7458842}. Despite this, in nonlinear systems, predictor feedback controllers still suffer from implementation limitations as the predictor itself is an implicit ordinary differential equation \cite{nikoBook}. To alleviate this issue, the book \cite{iassonBook} detailed a variety of numerical schemes combining finite differencing with successive approximations for approximating predictors for both linear and nonlinear systems with corresponding stability guarantees. However, \cite{iassonBook} does not explore the delay-adaptive case and thus no stability results exist for approximating the predictor when the delay is unknown. Furthermore, these schemes suffer from significant computational cost due to the requirement of a small discretization size. Therefore, \cite{bhan2024neuraloperatorspredictorfeedback} introduced operator learning to approximate the predictor mapping achieving $1000$x speedups compared to numerical solvers. However, \cite{bhan2024neuraloperatorspredictorfeedback} was for the simplest case of a nonlinear system - a constant known delay. Thus, we extend \cite{bhan2024neuraloperatorspredictorfeedback} to unknown delays. 

\subsection{Operator learning in control} \label{subsec:operator-learning}
We briefly discuss the advances of operator learning as a tool for real-time implementation of controllers. Neural operators were first introduced in \cite{392253} and then more recently popularized for learning PDE solutions in \cite{li2021fourier} and \cite{Lu2021}. The goal is to approximate any infinite dimensional operator between function spaces by finite dimensional neural networks. In control theory, operator learning was first introduced in \cite{bhan2023neural, KRSTIC2024111649} for approximating the gain kernel in linear PDEs, and has now expanded to output-feedback PDEs \cite{XIAO2024106620}, delayed PDEs systems \cite{10872816, QI2024105714}, gain-scheduling \cite{10918744} and even adaptive control \cite{lamarque2024adaptiveneuraloperatorbacksteppingcontrol,bhan2024adaptivecontrolreactiondiffusionpdes}. Notably, \cite{pmlr-v242-zhang24c} presented the first application of operator learning in real-world traffic scenarios, marking a significant milestone for practical deployment. We further review neural operators in Section \ref{sec:bg-neural-operators}.

\subsection{Notation}
For functions, $u: [0, D] \times \mathbb{R} \to \mathbb{R}$, we use $u_x(x, t) = \frac{\partial u}{\partial x}(x, t)$ and $u_t(x, t) = \frac{\partial u}{\partial t}(x, t)$ to denote derivatives. We use $C^1([t-D, t];\mathbb{R}^m)$ to denote the set of functions with continuous first derivatives mapping the interval $[t-D, t]$ to $\mathbb{R}^m$. For a $n$-vector, we use $|\cdot|$ for the Euclidean norm. For functions, we define the spatial $L^p$ norms as $\|u(t)\|_{L^p[0, D]} = (\int_0^D |u(x, t)|^p dx)^{\frac{1}{p}}$ for $p\in [1, \infty)$. 
\section{Technical background} \label{sec:background}
\subsection{Delay-adaptive control} \label{sec:bg-delay-adpative}
We study the nonlinear plant
\begin{align}
    \dot{X}(t) &= f(X(t), U(t-D))\,, \label{eq:main-problem}
\end{align}
where $X \in \mathbb{R}^n$, $U \in \mathbb{R}$, $f \in C^2(\mathbb{R}^n \times \mathbb{R}; \mathbb{R}^n)$ such that $f(0, 0) = 0$, and $D$ is an unknown delay within the interval $[\underline{D}, \overline{D}]$ where $\underline{D} > 0$. As standard in the predictor feedback literature \cite{delphine}, we introduce the following assumptions:
\begin{assumption} \label{assumption:strongly-forward-complete}
    The plant $\dot{X} = f(X, \Omega)$ with $\Omega$ scalar is strongly forward complete. 
\end{assumption}
\begin{assumption} \label{assumption:gas}
    There exists $\kappa \in C^2(\mathbb{R}^n; \mathbb{R})$ such that the feedback law  $U(t) = \kappa(X(t))$ guarantees that the delay-free plant is globally exponentially stable. 
\end{assumption}

Note that, under Assumption \ref{assumption:gas}, \cite[Theorem 4.14]{khalil} implies that there exists $\lambda > 0$ and a class $C^\infty$ radially unbounded positive definite function $V$ such that for any $X \in \mathbb{R}^n$, we have 
\begin{align}
    \frac{dV}{dX}(X) f(X, \kappa(X)) &\leq -\lambda V(X) \,, \label{eq:linear-lyapunov-1} \\
    |X|^2 &\leq V(X) \leq C_1 |X|^2\,,\label{eq:linear-lyapunov-2} \\ 
    \left|\frac{dV}{dX}(X) \right| &\leq C_2 |X| \label{eq:linear-lyapunov-3} \,, 
\end{align}
where $C_1, C_2 > 0$.

\begin{assumption}\label{assumption:lipschitz-dynamics}
Let $f(X, U)$ be the plant dynamics as in \eqref{eq:main-problem} and $\mathcal{X} \subset \mathbb{R}^n$ and $\mathcal{U} \subset \mathbb{R}$ be compact domains with bounds $\overline{X}$, $\overline{U}$ respectively. Then, there exists a constant $C_f(\overline{X}, \overline{U}) > 0$ such that $f$ satisfies the Lipschitz condition
\begin{align} \label{eq:lipschitz-dynamics-cond}
    |f(X_1, u_1) - f(X_2, u_2)| \leq C_f(|X_1-X_2| + |u_1-u_2|)\,,
\end{align}
for all $X_1, X_2 \in \mathcal{X}$ and $u_1, u_2 \in \mathcal{U}$. 
\end{assumption}

\begin{assumption} \label{assumption:growth-condition}
    The function $f$, the Jacobian matrix $\partial f/\partial X$, the control law $\kappa$ and its derivative satisfy the following growth conditions
    \begin{alignat}{2}
        |f(X, U)| &\leq M_1(|X|+|U|)\,,& \quad \left| \frac{\partial f}{\partial X}(X, U) \right| &\leq M_2 \,, \nonumber  \\
        |\kappa(X)| &\leq M_3|X|\,,& \quad \left| \frac{d\kappa}{dX}(X) \right| &\leq M_4\,,\nonumber 
    \end{alignat}
    for constants $M_1, M_2, M_3, M_4 > 0$.
\end{assumption}

Assumption \ref{assumption:strongly-forward-complete} ensures the plant \eqref{eq:main-problem} does not escape before the input reaches the system at $t=D$. Assumption \ref{assumption:gas} of global exponential stability is needed for the semi-global results in Section \ref{sec:theoretical-analysis}, though it can be relaxed for a local result \cite{delphine}. Assumption \ref{assumption:lipschitz-dynamics} is the only new assumption compared to \cite{delphine} and is required for approximating the predictor. Assumption \ref{assumption:growth-condition} is needed to show global asymptotic stability without any predictor approximation. Assumption \ref{assumption:growth-condition} makes Assumption \ref{assumption:strongly-forward-complete} redundant and is implied by Lipschitzness, but is included for clarity.

As standard in predictor feedback designs, we reformulate the plant \eqref{eq:main-problem} into the following ODE-PDE cascade where the delay is absorbed into the transport PDE with velocity $D$:
\begin{subequations}
\begin{align}
    \dot{X}(t) &= f(X(t), u(0, t))\,, \label{eq:plant-ode-pde-1} \\
    Du_t(x, t) &= u_x(x, t) \,, \label{eq:plant-ode-pde-2}\\ 
    u(1, t) &= U(t)\,.  \label{eq:plant-ode-pde-3}
\end{align}
\end{subequations}

The representation of the delay through a transport PDE was first introduced in \cite{miroslavDelay} where the analytical solution of the transport PDE recovers the original plant. Namely, we have that
\begin{align} \label{eq:transport-pde-sol}
    u(x, t) = U(t+D(x-1))\,, \quad  x \in [0, 1]\,, 
\end{align}
and thus when $x=0$ in the PDE representation, we recover $U(t-D)$ as in \eqref{eq:main-problem}. To compensate for the delay, we estimate the state $D$ timesteps in the future with the predictor given for all $x\in [0, 1]$ 
\begin{align} \label{eq:exact-predictor-no-adaptive}
    p(x, t) = X(t+Dx) = X(t) + D \int_0^x f(p(y, t), u(y, t))dy\,, 
\end{align}
which is itself, an implicit ODE whose analytical form is unknown for most nonlinear plant dynamics $f$. 

Then, under the assumption that the delay is known, the stabilizing control law is given by $U(t) = \kappa(p(1, t))$. When the delay is unknown, as in the problem setting of this work, the control law is chosen by 
\begin{align} \label{eq:exact-controller-adaptive}
    U(t) = \kappa(P(1, t))\,,
\end{align}
where the predictor is given by
\begin{align} \label{eq:exact-predictor-adaptive}
    P(x, t) = X(t) + \breve{D}(t) \int_0^x f(P(y, t), u(y, t)) dy\,. 
\end{align}
with the delay estimate $\breve{D}(t)$. 
Following \cite{delphine}, $\breve{D}(t)$ is updated via 
\begin{align} \label{eq:update-law-projector}
    \dot{\breve{D}}(t) &= \gamma \text{Proj}_{[\underline{D}, \overline{D}]} \left\{ \breve{D}(t), \phi(t) \right\}\,, \\ \label{eq:update-law}
    \phi(t) &= - \frac{\int_0^1 (1+x) q_1(x, t)w(x, t)dx}{1 + V(x) + b\int_0^1(1+x)w(x, t)^2dx}\,, 
\end{align}
where $b$ is a user-specified parameter, $V$ is the positive definite Lyapunov function in \eqref{eq:linear-lyapunov-1}, \eqref{eq:linear-lyapunov-2}, \eqref{eq:linear-lyapunov-3}, $w$ is given by the backstepping transformation
\begin{align} \label{eq:w-def-bcks-transform}
    w(x, t) = u(x, t) - \kappa(P(x, t))\,, 
\end{align}
and the scalar function $q_1$ is defined as 
\begin{align}
    q_1(x, t) = \frac{d \kappa}{d P}(P(x, t)) \Phi(x, 0, t) f(P(0, t), u(0, t))\,, 
\end{align}
where $\Phi$ is the transition matrix associated with the space-varying time-parameterized equation $(dr/dx)(x) = \breve{D}(t) (\partial f/\partial P) (P(x, t), u(x, t))r(x)$.

Then, under the control feedback law \eqref{eq:exact-controller-adaptive}, \cite[Theorem 1]{delphine} states that the plant is globally asymptotically stable. However, as mentioned, \eqref{eq:exact-predictor-adaptive} is not analytically known and therefore needs to be approximated for any real-world implementation. Thus, we introduce an approximation for the predictor $P$ and show that, under the universal approximation of neural operators, semi-global practical asymptotic stability of the closed-loop system is achieved. 

\subsection{Neural operators} \label{sec:bg-neural-operators}
For the readers reference, we briefly review the related background and theoretical results used in the analysis in Section \ref{sec:theoretical-analysis}. Neural operators are finite dimensional approximations of nonlinear operators (e.g. the solution operator to any nonlinear ODE) which map across function spaces. Thus, a neural operator takes a representation of the input function $c(x)$ and its evaluation point $x$ as inputs. Additionally, as input, it takes a value $y$, which specifies the point to evaluate the target function after applying the operator to $c$. The neural operator then provides an approximation of the value of the operator applied to $c(x)$ at the point $y$ in the target function space via finite neural networks.

We now describe neural operators from a rigorous perspective. Namely, let $\Omega_u \subset \mathbb{R}^{d_{u_1}}$, $\Omega_v \subset \mathbb{R}^{d_{v_1}}$ be bounded domains and let  $\mathcal{F}_c \subset C^0(\Omega_u; \mathbb{R}^c)$, $\mathcal{F}_v \subset C^0(\Omega_v; \mathbb{R}^v)$ be continuous function spaces. Then, we define a neural operator approximation for the nonlinear operator $\Psi$ as any  function satisfying the following form:
\begin{definition} \label{definition:neural-operator} \cite[Section 1.2]{lanthaler2024nonlocalitynonlinearityimpliesuniversality}
    Given the parameter for channel dimension $d_c$, we call any $\hat{\Psi}$ a neural operator given it satisfies the compositional form $\hat{\Psi} = \mathcal{Q} \circ \mathcal{L}_L \circ \cdots \circ \mathcal{L}_1 \circ \mathcal{R}$ where  $\mathcal{R}$ is a lifting layer, $\mathcal{L}_l, l=1,..., L$ are the hidden layers, and $\mathcal{Q}$ is a projection layer. That is, 
    $\mathcal{R}$ is given by 
    \begin{equation}
    \mathcal{R} : \mathcal{F}_c(\Omega_u; \mathbb{R}^c) \rightarrow \mathcal{F}_s(\Omega_s; \mathbb{R}^{d_c}), \quad c(x) \mapsto R(c(x), x)\,, 
\end{equation} where $\Omega_s \subset \mathbb{R}^{d_{s_1}}$, $\mathcal{F}_s(\Omega_s; \mathbb{R}^{d_c})$ is a Banach space for the hidden layers and $R: \mathbb{R}^c \times \Omega_u \rightarrow \mathbb{R}^{d_c}$ is a learnable neural network acting between finite-dimensional Euclidean spaces. For $l=1, ..., L$, each hidden layer is given by 
\begin{equation} \label{eq:generalNeuralOperator}
    (\mathcal{L}_l v)(x) := s \left( W_l v(x) + b_l + (\mathcal{K}_lv)(x)\right)\,, 
\end{equation}
where weights $W_l \in \mathbb{R}^{d_c \times d_c}$ and biases $b_l \in \mathbb{R}^{d_c}$ are learnable parameters, $s: \mathbb{R} \rightarrow \mathbb{R}$ is a smooth, infinitely differentiable activation function that acts component wise on inputs and $\mathcal{K}_l$ is the nonlocal operator given by 
\begin{equation} \label{eq:generalKernel}
    (\mathcal{K}_lv)(x) = \int_\mathcal{X} K_l(x, y) v(y) dy\,,
\end{equation}
where $K_l(x, y)$ is a kernel function containing learnable parameters. Lastly, the projection layer $\mathcal{Q}$ is given by 
\begin{equation}
    \mathcal{Q} : \mathcal{F}_s(\Omega_s; \mathbb{R}^{d_c}) \rightarrow \mathcal{F}_v(\Omega_v; \mathbb{R}^v), \quad s(x) \mapsto Q(s(x), y)\,, 
\end{equation}
where $Q$ is a finite dimensional neural network from $\mathbb{R}^{d_c} \times \Omega_v \rightarrow \mathbb{R}^v$. 
\end{definition}

This abstract formulation of neural operators covers a wide range of architectures including the Fourier Neural Operator (FNO) \cite{li2021fourier} as well as DeepONet \cite{Lu2021} where the differentiator between these approaches lies in the implementation of the kernel function $K_l$. However, in \cite{lanthaler2024nonlocalitynonlinearityimpliesuniversality}, the authors showed that a single hidden layer neural operator with a kernel given by the averaging kernel $K_l(x, y) =1/|\mathcal{X}|$ where $|\mathcal{X}|$ is the diameter of the domain is sufficient for universal approximation in the following sense
\begin{theorem}
    \label{thm:neural-operator-uat}
    \cite[Theorem 2.1]{lanthaler2024nonlocalitynonlinearityimpliesuniversality} Let $\Omega_u \subset \mathbb{R}^{d_{u_1}}$ and $\Omega_v \subset \mathbb{R}^{d_{v_1}}$ be two bounded domains with Lipschitz boundary. Let $\Psi: C^0(\overline{\Omega_u};\mathbb{R}^{d_{u_1}}) \rightarrow C^0(\overline{\Omega_v}; \mathbb{R}^{d_{v_1}}) $ be a continuous operator and fix a compact set $K \subset C^0(\overline{\Omega_u};\mathbb{R}^{d_{u_1}})$. Then for any $\epsilon > 0$, there exists a channel dimension $d_c > 0$ such that a single hidden layer neural operator $\hat{\Psi}: K \rightarrow C^0(\overline{\Omega_v}; \mathbb{R}^{d_{v_1}})$ with kernel function $K_l(x, y) =1/|\mathcal{X}|$ satisfies  
\begin{equation}
    \sup_{u \in K} |\Psi(u)(y) - \hat{\Psi}(u)(y)|\leq \epsilon\,,
\end{equation}
for all values $y \in \Omega_v$.
\end{theorem}

The key corollary to Theorem 1 is that, many of the architectures aforementioned (e.g. FNO, Laplace neural operator, DeepONet) contain kernel functions $K_l$ that recover the averaging kernel and thus are universal under the setting of Theorem \ref{thm:neural-operator-uat}.  
Henceforth, for the remainder of this study, we assume that any reference to a neural operator is an architecture of the form given in Definition \ref{definition:neural-operator} with a viable kernel function such that Theorem \ref{thm:neural-operator-uat} holds. We are now ready to introduce the predictor operator to be approximated.

\section{Neural operator approximate predictors} \label{sec:theoretical-analysis}
\begin{definition} \label{defintion:predictor-operator}
Let $X \in \mathbb{R}^n$, $U \in C^2([0, 1]; \mathbb{R})$, $\varphi \in \mathbb{R}^+$. Then, we define the \textbf{predictor operator} as the mapping $\mathcal{P}: (X, U, \varphi) \rightarrow P$ where $P(s) = \mathcal{P}(X, U, \varphi)(s)$ satisfies for all $s \in [0, 1]$, 
    \begin{align} \label{eq:predictor-operator}
        P(s) - X - \varphi \int_0^s f(P(s), U(s)) ds = 0\,. 
    \end{align}
\end{definition}

Notice that, by definition the predictor operator yields the solution to \eqref{eq:exact-predictor-adaptive} where $\varphi$, as an input to the operator, is the estimated delay. Furthermore, in order to approximate this mapping with a neural operator such that Theorem \ref{thm:neural-operator-uat} holds, we require that the predictor operator is continuous. Thus, we prove the following Lemma:
\begin{lemma}
    \label{lemma:continuity-of-predictor}
    Let Assumption \ref{assumption:lipschitz-dynamics} hold. Then, for any $X_1, X_2 \in \mathcal{X}$, $U_1, U_2 \in C^2([0, 1]; \mathcal{U})$, and $\varphi_1, \varphi_2 \in (\underline{D}, \overline{D})$ the predictor operator $\mathcal{P}$ given in Definition \ref{defintion:predictor-operator} satisfies 
    \begin{align} \label{eq:predictor-lipschitz-bound}
        \|\mathcal{P}(&X_1, U_1, \varphi_1) - \mathcal{P}(X_2, U_2, \varphi_2) \|_{L^\infty[0, 1]} \nonumber \\   &\leq C_{\mathcal{P}} \left(|X_1-X_2| + \|U_1 - U_2\|_{L^\infty[0, 1]} + |\varphi_1 - \varphi_2| \right)\,,
    \end{align}
    with Lipschitz constant
    \begin{align}
        C_{\mathcal{P}} &= e^{\overline{D}C_f} \max\left\{1, \Xi, \overline{D}C_f\right\}\,, \\
        \Xi &= C_f\left[\overline{U} + e^{\overline{D}C_f}(\overline{X} + C_f\overline{D}\overline{U})\right]\,. 
    \end{align}
\end{lemma}
\begin{proof}
    For all $s \in [0, 1]$, let $\bar{P}_1(s) := \mathcal{P}(X_1, U_1, \varphi_1)$ and likewise  $\bar{P}_2(s) := \mathcal{P}(X_2, U_2, \varphi_2)$. Note that, for $s\in [0,1]$ the predictor is uniformly bounded. First, by definition and Assumption \ref{assumption:lipschitz-dynamics}, we have 
    \begin{align}
        \bar{P}_1(s) &= X_1 + \varphi_1 \int_0^s f(\bar{P}_1(\theta), U_1(\theta)) d \theta \nonumber \\ 
        &\leftstackrel{\eqref{eq:lipschitz-dynamics-cond}}{\leq} X_1 + \varphi_1 \int_0^s C_f(|\bar{P}_1(\theta)| + |U_1(\theta)|) d\theta \nonumber  \\ 
        &\leq X_1 + \varphi_1 C_f \|U_1\|_{L^\infty[0, 1]} + \varphi_1 C_f \int_0^1 |\bar{P}_1(\theta)| d \theta 
    \end{align}
    Applying Gronwall's inequality yields 
    \begin{align} \label{eq:predictor-uniform-bound}
        |\bar{P}_1(s)| \leq e^{\overline{D}C_f}(\overline{X} + C_f \overline{D} \phantom{,} \overline{U})\,. 
    \end{align}
    Then, applying \eqref{eq:predictor-uniform-bound} yields the following calculation
       \begin{alignat}{2}
        \bar{P}_1(s) - \bar{P}_2(s) &=&& X_1 - X_2 + \varphi_1 \int_{-1}^s f(\bar{P}_1(\theta ), U_1(\theta )) d \theta  \nonumber  \\ & &&\nonumber  - \varphi_2 \int_{-1}^s f(\bar{P}_2(\theta ), U_2(\theta )) d\theta \\ \nonumber
        &\leq&& |X_1-X_2| \nonumber \\\nonumber & && + (\varphi_1-\varphi_2) \int_{-1}^s  f(\bar{P}_1(\theta ), U_1(\theta )) d\theta \\\nonumber & && + \varphi_2 \int_{-1}^s\bigg( f(\bar{P}_1(\theta), U_1(\theta )) \\ \nonumber  & &&- f(\bar{P}_2(\theta), U_2(\theta )) \bigg) d \theta   \\ 
        &\leftstackrel{\eqref{eq:lipschitz-dynamics-cond}}{\leq}&& |X_1 - X_2| + (\varphi_1-\varphi_2) C_f \|U_1\|_{L^\infty[0, s]} \nonumber \\ \nonumber & &&  + (\varphi_1-\varphi_2) C_f \int_{-1}^s |P(\theta )| d \theta \\ & && + \varphi_2C_f \int_{-1}^s |\bar{P}_1(\theta) - P_2(\theta)|\nonumber \\ & && \nonumber + |U_1(\theta)-U_2(\theta)| d \theta \\ 
        &\leftstackrel{\eqref{eq:predictor-uniform-bound}}{\leq}&&  |X_1 - X_2| \nonumber \\ & &&  + |\varphi_1-\varphi_2| C_f \bigg[\overline{U}  + e^{C_f\overline{D}}(\overline{X} + C_f \overline{D} \phantom{,} \overline{U})\bigg] \nonumber \\ & && + \overline{D}C_f \|U_1-U_2\|_{L^\infty[-1, s]} \nonumber \\ & && 
        + \overline{D} C_f \int_{-1}^s |\bar{P}_1(\theta) - \bar{P}_2(\theta)|d\theta \,.
    \end{alignat}
    Applying Gronwall's again yields the desired result. 
\end{proof}

Given the continuity in Lemma \ref{lemma:continuity-of-predictor}, we can then apply Theorem \ref{thm:neural-operator-uat} to the predictor operator yielding the existence of an arbitrarily close neural operator approximation:
\begin{theorem} \label{thm:uat-predictor}
    Fix a compact set $K \subset \mathcal{X} \times C^2([0, 1]; \mathcal{U}) \times \mathcal{D}$. Then, for all $\overline{X}$, $\overline{U}$, $\overline{D}$, $\epsilon > 0$, there exists a neural operator approximation $\hat{\mathcal{P}} : K \to C^1([0, 1]; \mathbb{R}^n)$ such that 
    \begin{align}
        \sup_{(X, U, \varphi) \in K} &|\mathcal{P}(X, U, \varphi)(s)  - \hat{\mathcal{P}}(X, U, \varphi)(s) |  < \epsilon \,, 
    \end{align}
    for all $s \in [0, 1]$.
\end{theorem}

Now, we have established the existence of an arbitrarily close approximate predictor, but we did not discuss the detail on the number of parameters or data required to obtain this bound. Such analysis is beyond the goal of this paper, but we refer the reader to \cite{mukherjee2024size}, \cite{NLM2024data}, \cite{LST2024discretization} for further details. Further, given the approximation is arbitrarily close in $\epsilon$, this perturbation will affect the stability of the system compared to the global asymptotic result conducted with the exact predictor (which can never be implemented in practice). Therefore, in the next section, we identify the exact effect of the approximate predictor on the overall feedback loop. 
\section{Stability analysis under approximate predictors} \label{sec:stability}
We are now ready to analyze the plant \eqref{eq:plant-ode-pde-1}, \eqref{eq:plant-ode-pde-2}, \eqref{eq:plant-ode-pde-3} under the controller with the neural operator approximated predictor 
\begin{align} \label{eq:approximate-predictor-feedback-law}
    U(t) =& \kappa\left(\hat{P}\left(X(t), T_{D}(t)U, \breve{D}(t)\right)\right)\,,
\end{align}
where $T_{D}(t)$ is the historical control operator which applied to a function $U$ yields
\begin{align} 
\label{eq:shift-operator-control-history}
    (T_{\breve{D}}(t)U) (x) := U(t-D(t)(x -1))\,,  \quad  x \in [0, 1)\,.
\end{align} 
The operator $T_{D}(t)$ shifts the function $U$ back by $t-D(x-1)$ units such that, we have $u(x, t) = (T_{D}(t)U)(x)$ and therefore the resulting output of the exact predictor operator satisfies \eqref{eq:exact-predictor-adaptive}. Notice that, in this case, we consider the full actuator measurement such that $u(x, t)$ is always known which is not necessarily guaranteed for every application. However, as shown in Section \ref{sec:numerical}, this is a reasonable assumption for biological systems of which we study in this work. 

We are now ready to present the main result of the paper. 
\begin{theorem} \label{thm:main-result}
    Let the system \eqref{eq:plant-ode-pde-1}, \eqref{eq:plant-ode-pde-2}, \eqref{eq:plant-ode-pde-3} satisfy Assumptions \ref{assumption:strongly-forward-complete}, \ref{assumption:gas}, \ref{assumption:lipschitz-dynamics}, \ref{assumption:growth-condition}. Define the functional
    \begin{align}
        \Gamma(t) = |X(t) |^2 + \int_{t-D}^t U(\theta)^2d \theta + |\tilde{D}(t)|^2\,,  
    \end{align}
    where $\tilde{D}(t) = D-\breve{D}(t)$. Then, there exists constants $\gamma^\ast$, $b^\ast(\bar{D}, C_f)$, $\overline{\Gamma}(\overline{X}, \overline{U}, \overline{D}) > 0$, $\beta_1^\ast \in \mathcal{KL}$ and class $\mathcal{K}_\infty$ functions $\alpha_1^\ast, \alpha_2^\ast, \alpha_3^\ast, \alpha_4^\ast, \alpha_5^\ast$ such that if $\gamma < \gamma^\ast$, $b > b^\ast$, $\epsilon < \epsilon^\ast$ where
    \begin{align} \label{eq:epsilon-star}
        \epsilon^\ast = (\alpha_1^\ast)^{-1}(\overline{\Gamma
     }-\alpha_2^\ast(\overline{\Delta D}))\,, 
    \end{align}
    where $\overline{\Delta D} = \overline{D}- \underline{D}$ and the initial state is constrained to
    \begin{align} \label{eq:gamma-zero}
    \Gamma(0) \leq \alpha_3^\ast(\overline{\Gamma} - \alpha_1^\ast(\epsilon) - \alpha_2^\ast(\overline{\Delta D}))\,, 
    \end{align}
    then, 
    \begin{align} \label{eq:gamma-t}
        \Gamma(t) \leq \alpha_3^\ast(\Gamma(0)) + \alpha_1^\ast(\epsilon) + \alpha_2^\ast(\overline{\Delta D})\,, 
    \end{align}
    and 
    \begin{align} \label{eq:regulation}
       |X(t)|^2 &\leq \beta_1^\ast(|X(0)|^2, t) + \alpha_4^\ast(\overline{\Delta D}) + \alpha_5^\ast (\epsilon)\,,  \\ 
        \|u(t)\|_{L^2[0, D]}^2 &\leq \beta_1^\ast(\|u(0)\|_{L^2[-D, 0]}^2, t) + \alpha_4^\ast(\overline{\Delta D}) + \alpha_5^\ast(\epsilon)\,.   
    \end{align}
\end{theorem}

Notice that Theorem \ref{thm:main-result} is weaker than the result in \cite[Theorem 7]{bhan2023neural} because it depends on the delay projection bounds $\overline{D}$, $\underline{D}$. This is expected since the error introduced by the predictor is additive, leading to complications similar to those in robust adaptive control, where regulation depends on projector bounds (See \cite{IKHOUANE1998429}).

Additionally, it may seem counterintuitive that $\epsilon^\ast$ increases with $\overline{\Gamma}$, but this is expected: a larger radius of states allows for a larger transient, thus loosening the $\epsilon^\ast$ bound. 

Finally, we briefly compare Theorem 3 with other adaptive control implementations using neural operators in PDEs. In \cite[Theorem 4]{bhan2024adaptivecontrolreactiondiffusionpdes} and \cite[Theorem 4]{lamarque2024adaptiveneuraloperatorbacksteppingcontrol}, the authors approximate the gain kernel (not the controller), which introduces a multiplicative error and leads to global stability. However, this requires approximating the time derivative of the update parameter. In contrast, in Theorem \ref{thm:main-result}, our approximation of $P$ is similar to approximating the direct control law, as in \cite[Theorem 3]{bhan2023neural} resulting in an additive error and yielding a semiglobal practical result, without the need to approximate the time derivative of $P_t$.

\begin{proof}
    We first introduce a backstepping transform, that under the feedback law \eqref{eq:approximate-predictor-feedback-law}, transforms the system into a target system that we then use for our stability analysis. 
    \begin{lemma} \label{lemma:backstepping-lemma}
    Under the backstepping transformation with the \emph{exact predictor} given by 
    \begin{align} \label{eq:lemma2-bcks-transform}
        w(x, t) = u(x, t)- \kappa(P(x, t))\,,
    \end{align}
    the system \eqref{eq:plant-ode-pde-1}, \eqref{eq:plant-ode-pde-2}, \eqref{eq:plant-ode-pde-3} with control law given by \eqref{eq:approximate-predictor-feedback-law}
    \begin{subequations} is transformed into    
    \begin{align}
        \dot{X}(t) &= f(X(t), \kappa(X(t)) + w(0, t))\,, \label{eq:plant-target-ode-pde-1} \\ 
        Dw_t &= w_x - \tilde{D}(t) q_1(x, t) - D \dot{\breve{D}}(t) q_2(x, t)\,,\label{eq:plant-target-ode-pde-2}  \\
        w(1, t) &= \kappa(P(x, t)) - \kappa(\hat{p}(x, t))\,, \label{eq:plant-target-ode-pde-3} 
    \end{align}
    \end{subequations}
    where $\tilde{D} = D-\breve{D}(t)$, $\hat{p}(x, t) = P_1(x)$, and $q_2$ is given by
    \begin{align} \label{eq:q2}
        q_2 =& \frac{\partial \kappa}{\partial P}(P(x, t)) \int_0^x \Phi(x, y, t) f(P(y, t), \kappa(P(y, t)) \nonumber \\ &+w(y, t)) dy\,. 
    \end{align}
    \end{lemma}
    Since we used the exact backstepping transform in \eqref{eq:lemma2-bcks-transform}, the proof of Lemma \ref{lemma:backstepping-lemma} is exactly the same as the proof of \cite[Lemma 1]{delphine} except the boundary term at $x=1$ which is obtained via direct substitution of $U(t) =\kappa(\hat{P}(\cdot)) - \kappa(P(\cdot))$ in \eqref{eq:lemma2-bcks-transform}.  

    To analyze \eqref{eq:plant-target-ode-pde-1}, \eqref{eq:plant-target-ode-pde-2}, \eqref{eq:plant-target-ode-pde-3}, we introduce the following Lyapunov-Krasovskii functional
    \begin{align}
        W(t) &= D\log N(t) + \frac{b}{\gamma}\tilde{D}(t)^2  \\ 
        N(t) &= 1+V(x) + b\int_0^1(1+x)w(x, t)^2 dx\,.
    \end{align}
    Taking the time derivative and substituting $w_t$ yields
    \begin{align}
        \dot{W}(t) =& \frac{1}{N(t)} \bigg( D \frac{\partial V(X)}{\partial X} \bigg[f(X(t), \kappa(X(t)) + w(0, t)) \nonumber  \bigg] \bigg) \nonumber \\ & + \frac{2b}{N(t)} \int_0^1 (1+x) w(x, t)\bigg(w_x(x, t) - \tilde{D}(t) q_1(x, t) \nonumber \\ &- D \dot{\breve{D}}(t) q_2(x, t) \bigg) dx - \frac{2b}{\gamma} \tilde{D} \dot{\breve{D}}(t) \,.
    \end{align}
    Using the Lipschitzness of $f$ in Assumption \ref{assumption:lipschitz-dynamics} and the Lyapunov condition in Assumption \ref{assumption:gas}, we have
    \begin{align}
        \dot{W}(t) &\leq  \frac{1}{N(t)} \left(-D \lambda |X(t)|^2 + DC_2 C_f |X(t)| |w(0, t)| \right) \nonumber \\  & + \frac{2b}{N(t)} \int_0^1 (1+x) w(x, t)\bigg(w_x(x, t) - \tilde{D}(t) q_1(x, t) \nonumber \\ &- D \dot{\breve{D}}(t) q_2(x, t) \bigg) dx - \frac{2b}{\gamma} \tilde{D} \dot{\breve{D}}\,.
    \end{align}
    Distributing the second term, applying integration by parts, and substituting \eqref{eq:update-law} yields
    \begin{align}
        \dot{W}(t)  \leq& \nonumber  \frac{1}{N(t)} \bigg( -D \lambda |X(t)|^2 + DC_2C_f |X(t)||w(0, t)| \nonumber \\  &+ 2bw^2(1, t) - 2bw^2(0, t)  - 2b\|w(t)\|_2^2 \nonumber \\ &+ 2bD \dot{\breve{D}}(t) \int_0^1 (1+x)w(x, t) q_2(x, t) dx  \bigg) \nonumber \\ &+ \frac{2b}{\gamma}\tilde{D}(t)(\gamma \phi(t) - \dot{\breve{D}}(t))\,.
    \end{align}
    Using \cite[Lemma 2 and 3]{delphine} which states that there exists $M_5, M_6, M_7 > 0$ such that 
    \begin{align}
    |P(x, t)| &\leq M_5(|X| + \|w(t)\|_2)\,, \label{eq:breve-preidctor-bound} \\ 
    2bD \left|\int_0^1 (1+x)q_2(x, t) w(x, t) dx \right| &\leq M_6\left(|X|^2 + \|w(t)\|_2^2 \right)\,, \\
    \left| \dot{\breve{D}}(t) \right| &\leq \gamma M_7\,, 
\end{align}
in conjunction with Young's inequality yields
\begin{align}
    \dot{W}(t) \leq & \nonumber \frac{1}{N(t)} \bigg(-\eta \left(|X(t)|^2 + \|w(t)\|_2^2 \right) \\ \nonumber  &- \bigg(2b- \frac{C_2^2 D C_f^2}{2 \lambda}\bigg) w^2(0, t) \\ &+ \gamma M_6 M_7 \bigg(|X(t)|^2 + \|w(t)\|_2^2 \bigg)  \nonumber  \\ & + 2b w^2(1, t)\bigg)\,,
\end{align}
where $\eta := \min(D\lambda/2, 2b)$. Thus, given that 
\begin{align}
    b &> \frac{C_2^2 \overline{D} C_f^2 }{4\lambda} =: b^\ast\,,  \\
    \gamma &<  \frac{\eta}{M_6M_7} =: \gamma^*\,, 
\end{align}
guarantees there exists $C>0$ such that
\begin{align} \label{eq:w-dot}
    \dot{W}(t) \leq& \nonumber  -\frac{C}{N(t)}(|X(t)|^2 + \|w(t)\|^2_2) \\ & + \frac{2b}{N(t)}(\kappa(P(x, t)) - \kappa(\hat{p}(x, t)))^2 \,.
\end{align}
Now, using the definition of $W(t)$ and Assumption \ref{assumption:gas}, we have that 
\begin{align}
    e^{W(t)/D - \frac{b}{\gamma{D}}(\tilde{D}(t))^2} - 1 \leq& \max\{C_1, 2b\}(|X(t)|^2 + \|w(t)\|_2^2)\,,
\end{align}
Using reverse Young's inequality (See \ref{appendix:reverse-youngs}), we have
\begin{align}
    e^{W(t)/D}e^{-\frac{b}{\gamma{D}}(\tilde{D}(t))^2} \geq 2e^{W(t)/(2D)}-e^{\frac{b}
    {\gamma D}{\tilde{D}(t))^2}} 
\end{align}
yielding
\begin{align}
 2e^{W(t)/D} - &e^{\frac{b}{\gamma D}(\tilde{D}(t))^2}-1 \nonumber \\ 
 &\leq \max\{C_1, 2b\}(|X(t)|^2 + \|w(t)\|_2^2)\,, 
\end{align}
which then results in 
\begin{align}
   - (|X(t)|^2 + \|w(t)\|_2^2) \leq&  \frac{1}{\max\{C_1, 2b\}} \bigg(-2e^{W(t)/D} \nonumber \\ &+ e^{\frac{b}{\gamma D}(\tilde{D}(t))^2} + 1 \bigg)\label{eq:w-lowerbound}\,. 
\end{align}
Substituting into \eqref{eq:w-dot} yields
\begin{align}
    \dot{W}(t) \leq& -\frac{C}{\max\{C_1, 2b\}N(t)}   \nonumber 2e^{\frac{W(t)}{D}} \\ &+ \bigg[ \frac{C}{\max\{C_1, 2b\}N(t)}\left(1+e^{\frac{b}{D\gamma}(\overline{\Delta D})^2}\right) \nonumber \\  &+ b(\kappa(P(x, t)) - \kappa(\hat{p}(x, t)))^2 \bigg]\,. 
\end{align}
Then, by Assumption \ref{assumption:growth-condition}, we have that for all $x \in [0, 1]$, $t\geq 0$
\begin{align}
    (\kappa(P(x, t)) - \kappa(\hat{p}(x, t)))^2 \leq& M_3^2 |P(x, t) - \hat{p}(x, t)|^2 \nonumber \\ \leq& M_3^2 \epsilon^2\,.
\end{align}
Following 
\cite[Theorem C.3]{kkk}, we have that there exists functions $\beta_1 \in \mathcal{KL}$ and $\alpha_1, \alpha_2, \alpha_3\in \mathcal{K}_\infty$ such that 
\begin{align}
    W(t) \leq& \beta_1(W(0), t) +\nonumber  \alpha_1(\overline{\Delta D})  \\ &+  \alpha_2(\sup_{0 \leq \tau \leq t}(\kappa(P(x, t)) - \kappa(\hat{p}(x, t)))^2 ) \nonumber \\ 
    \leq& \beta_1(W(0), t) + \alpha_1(\overline{\Delta D}) + \alpha_3(\epsilon)\label{eq:w-stability-estimate}\,,
\end{align}
when $\epsilon$ and $\overline{\Delta D}$ is small enough relative to all possible values of $W$ (i.e. $\epsilon + (\overline{\Delta D}) \leq \alpha_4(W(t)) \leq \alpha_5(\overline{X}+\overline{U}+\overline{D})$ where $\alpha_4, \alpha_5 \in \mathcal{K}_\infty$). 
To obtain an estimate on $\Gamma$, note that from \cite[Eqn. (47), (48)]{delphine}, we have that 
\begin{align}
    \Gamma(t) &\leq \left(D\left(r_1 + \frac{r_2}{b} \right) + 1 + \frac{\gamma D}{b} \right) \left(e^{\frac{W(t)}{D}} - 1 \right) \,, \\
    W(t) &\leq D \left(c_1 + 2b \left(s_1 + \frac{s_2}{D} \right) + \frac{b}{\gamma D} \right) \Gamma(t)\,.
\end{align}
Therefore, substituting in the stability estimate in \eqref{eq:w-stability-estimate}, we obtain the exists of class $K^\infty$ functions $\alpha_1^\ast$, $\alpha_2^\ast$, $\alpha_3^\ast$ such that 
\begin{align}
    \Gamma(t) \leq \alpha^\ast_1(\Gamma(0)) + \alpha_2^\ast(\epsilon)+\alpha_3^\ast(\overline{\Delta D}) \,,
\end{align}
when $\Gamma(0) \leq (\alpha^\ast_1)^{-1}(\overline{\Gamma} - \alpha_2^\ast(\epsilon) - \alpha_3^\ast(\overline{\Delta D}))$ where $\overline{\Gamma} = \overline{X}+ \overline{U} + \overline{D}$. 

To show obtain a bound on $X(t)$, note that we have
\begin{align}
    |X(t)|^2 \leq V(t) \leq (e^{\frac{W(t)}{D}}-1)\, \label{eq:x-upperbound}.
\end{align}
Using the stability estimate on $W(t)$ in \eqref{eq:w-stability-estimate}, along with the inequalities \eqref{eq:w-lowerbound}, \eqref{eq:x-upperbound} we have that there exists $\beta_2 \in \mathcal{KL}$, $\alpha_6, \alpha_7 \in \mathcal{K}_\infty$ such that 
\begin{align}
    |X(t)|^2 \leq \beta_2(|X(0)|^2, t) + \alpha_6(\overline{\Delta D}) + \alpha_7 (\epsilon)\,. 
\end{align}

Similarly, to show a stability estimate on $U(t)$, notice that \eqref{eq:breve-preidctor-bound} and Assumption \ref{assumption:growth-condition} implies that there exists constants $r_1, r_2, s_1, s_2 > 0$ such that
\begin{align}
    \|u(t)\|_2^2 \leq& r_1|X(t)|^2 + r_2\|w(t)\|^2\,,  \label{eq:u-xw-bound} \\
    \|w(t)\|_2^2 \leq& s_1|X(t)|^2 + s_2\|u(t)\|^2\,.\label{eq:w-ux-bound}
\end{align}
Thus, using \eqref{eq:u-xw-bound}, \eqref{eq:w-ux-bound}, along with the definition of $W(t)$, we have that 
\begin{align}
    \int_{t-D}^t U(s)^2ds = D\|u(t)\|_2^2 &\leq D(r_1 |X(t)|^2 + r_2\|w(t)\|_2^2)\nonumber  \\
    &\leq D\left(r_1 + \frac{r_2}{b}\right) \left(e^{\frac{W(t)}{D}} - 1\right)\,. \label{eq:u-upper-bound}
\end{align}
Then, repeating the argument above with the stability estimate in \eqref{eq:w-stability-estimate} along with the inequalities \eqref{eq:w-lowerbound}, \eqref{eq:u-upper-bound}, we have there exists $\beta_3 \in\mathcal{KL}$, $\alpha_8, \alpha_9 \in \mathcal{K}_\infty$ such that 
\begin{align}
    \|u(t)\|_{L^2[0, D]}^2 \leq \beta_3(\|u(0)\|_{L^2[-D, 0]^2}, t) + \alpha_8(\overline{\Delta D}) + \alpha_9(\epsilon)\,.   
\end{align}
Letting $\beta_1^\ast(\cdot, \cdot) = \max\{\beta_2(\cdot, \cdot), \beta_3(\cdot, \cdot)\}$, $\alpha_4^\ast(\cdot) = \max\{\alpha_6(\cdot), \alpha_8(\cdot)\}$ and $\alpha_5^\ast(\cdot) = \max\{\alpha_7(\cdot), \alpha_9(\cdot)\}$ completes the result.

\end{proof}
\section{Numerical Experiments} \label{sec:numerical}

\begin{figure*}[!htbp]
    \centering
    \includegraphics[trim=0 0.65cm 0 0, clip]{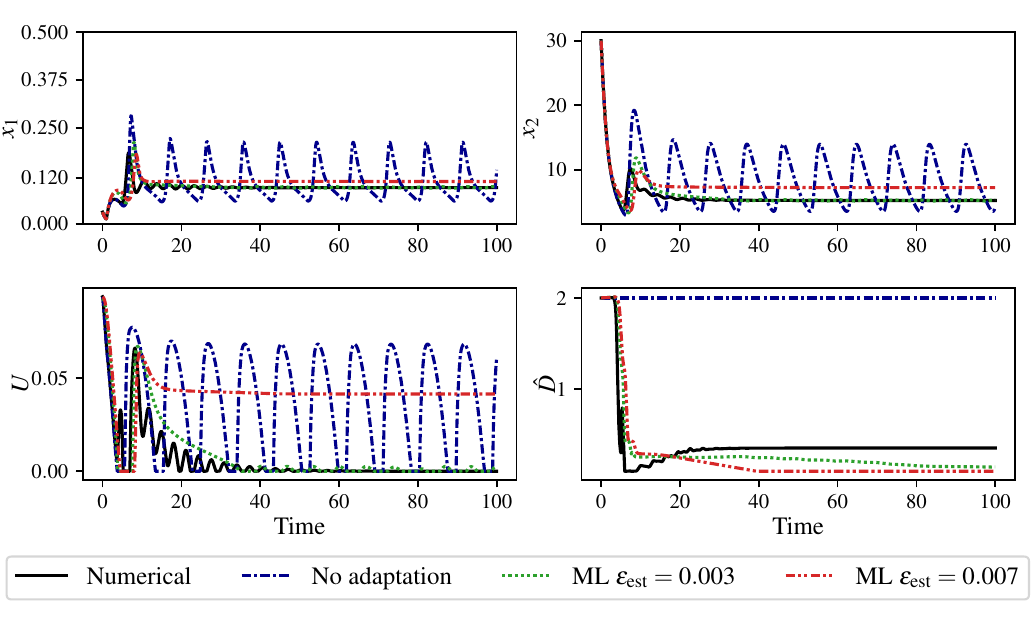}
    \caption{Simulation of the plant \eqref{eq:main-problem} with various approximate predictors. The initial state is $X(0) = [0.03, 30]$, delay is $D=1$, $\hat{D}(0) = 2$, $\gamma=1000$, $b=1$. The black line indicates the numerical predictor, the blue-line is the predictor without delay-adaptation, the red line indicates the DeepONet predictor with higher error and the green line indicates the DeepONet predictor trained to optimality.}
    \label{fig:main-fig}
\end{figure*}

\color{blue}
For the simulation analysis, we provide all code, numerical parameters, and datasets publicly on Github (\url{https://github.com/lukebhan/Neural-Operator-Delay-Adaptive-Predictor-Feedback}).
\color{black}

For numerical validation, we consider the biological system in \cite{4282275} consisting of two proteins. The activator protein that promotes expression of itself and a repressor protein that represses the expression of the activator. Such biological clocks play fundamental roles in cell physiology \cite{Pomerening2003} and control of such systems is of valuable interest to synthetic biologists for the design of new medicines. The system dynamics manifest as 
\begin{align}
    \dot{x}_1 &= -x_1 + f_1(x_1, x_2) + U(t-D)\,, \\
    \dot{x}_2 &= -\frac{x_2}{2} + f_2(x_1)\,,
\end{align}
where $x_1$ is the concentration of the activator protein, $x_2$ is the concentration of the repressor protein and $f_1$, $f_2$ are Hill functions given by 
\begin{align}
    f_1(x_1, x_2) &= \frac{K_1x_1^2 + K_a}{1+x_1^2 + x_2^2}\,, \\
    f_2(x_1) &= \frac{K_2x_1^2 + K_b}{1+x_1^2}\,,
\end{align}
with $K_1=K_2 = 300, K_a=0.04, K_b=0.004$. 
The input $U$ controls the activator protein concentration with a delay due to the time to pass through an inlet pipe. Additionally, we assume the pipe is equipped with a UV spectrometer such that we can measure the presence of the activator protein along the pipe and thus the distributed input $u(\cdot, t)$ is known. The goal is to stabilize the system to the unstable target equilibrium $(x_1^\ast, x_2^\ast)$. 

A control law satisfying Assumption \ref{assumption:gas} is given by
\begin{align}
    \kappa(X(t)) = -f_1(x_1, x_2) + f_1(x^\ast_1, x_2^\ast)\,,
\end{align}
where $(x_1^\ast, x_2^\ast) = (0.0939, 5.2525)$ is an unstable equilibrium setpoint of the system. The associated Lyapunov function is given by $V = (X-X^\ast)^T(X-X^\ast)$ and note the system experiences a limit cycle \cite[Figure 2]{delphine} when an open-loop control law $U(\cdot) = 0$ is used.  

To train a neural operator approximate predictor, we generate a dataset consisting of $5000$ instances where each instance is sampled from simulating the system with a numerical implementation of the predictor \cite{iassonBook} with various initial conditions and initial delay estimations (See code for parameters; Dataset generation took $60$ minutes). We then train two different neural operator architectures (DeepONet,FNO) requiring approximately $10$ minutes on a Nvidia $3090$Ti GPU. Both architectures performed similarly, but the speedup of DeepONet was larger as given in Table \ref{tab:comp-time}.

We present the simulation results of the dynamical system in Figure \ref{fig:main-fig} with the DeepONet predictor. To highlight the effect of $\epsilon$, we consider two approximate predictors where one has the training stopped early and thus contains a large $\epsilon$-error and the other is trained to optimality. In the case of early stopping, the system converges around the equilibrium, but with a larger radius then the  predictor trained to optimality. Furthermore, despite the smaller error, the predictor trained to optimality does not fully converge to the setpoint, but converges to the point $(0.094, 5.39)$ (Recall $(x_1^\ast, x_2^\ast) = (0.0939, 5.2525)$) due to the uniform approximation error $\epsilon$. In both cases, we see that the effect of the training error $\epsilon$ directly impacts the radius of convergence around the target equilibrium setpoint as expected given in Theorem \ref{thm:main-result}. Furthermore, note that the delays do not converge to $D$, but close enough to ensure stability. This is a well-known phenomena of adaptive control where the global goal of the scheme \emph{is not parameter convergence}, but is instead stability \cite{kkk}. 

Lastly, we investigated the speedup of the neural operator approximate predictors. At all discretizations, the ML predictor is faster then the numerical approach with the largest speedup of $15$x under the smallest discretization size ($dx=0.001$). We note that the discretization size must be less then $0.005$ numerical stabilization (Using an Euler approximation scheme). This is less then the $100$x speedup achieved in \cite[Table 1]{bhan2024neuraloperatorspredictorfeedback}, but is attributed to the fact the dynamics $f$ are much simpler to simulate in this example compared to the robotic manipulator in \cite{bhan2024neuraloperatorspredictorfeedback}. If one applied the same methodology with more expensive forward dynamics $f$, one should except to see even larger speedup over the numerical method under a neural operator approximate predictor.  
\begin{table}[]
\centering
\caption{Computation time for various approximate predictors averaged over 1000 samples (seconds). }
\label{tab:comp-time}
\resizebox{\linewidth}{!}{%
\begin{tabular}{l|ccccc}
\begin{tabular}[c]{@{}l@{}}Step \\ size (dx)\end{tabular} & \multicolumn{1}{l}{Numerical $\downarrow$} & \multicolumn{1}{l}{DeepONet $\downarrow$} & \multicolumn{1}{l}{FNO $\downarrow$} & \multicolumn{1}{l}{\begin{tabular}[c]{@{}l@{}}DeepONet\\  Speedup $\uparrow$\end{tabular}} & \multicolumn{1}{l}{\begin{tabular}[c]{@{}l@{}}FNO\\ Speedup $\uparrow$\end{tabular}} \\ \hline
0.01                                                      & 1.601                                      & \textbf{0.496}                            & 1.331                                & \textbf{3.22x}                                                                             & 1.20x                                                                                \\
0.005                                                     & 3.295                                      & \textbf{0.587}                            & 1.440                                & \textbf{5.61x}                                                                             & 2.29x                                                                                \\
0.001                                                     & 18.197                                     & \textbf{1.212}                            & 2.108                                & \textbf{15.01x}                                                                            & 8.63x                                                                               
\end{tabular}%
}
\end{table}

\section{Conclusion}
In this paper, we presented the first result studying neural operator approximate predictors when the delay is \emph{unknown} and therefore estimated online. In particular, we formalized and proved the existence of an arbitrary close neural operator approximation of the predictor operator which was then used in to prove semi-global $\epsilon$-practical asymptotic stability via a Lyapunov-Kasovskii functional. We emphasize that our theoretical result is applicable to any \emph{black-box} approximate predictor satisfying a uniform error bound. Lastly, we concluded by training a neural operator predictor to compensate for a input-lag in a biological protein activator/repressor system achieving speedups on the magnitude of $15$x. In the future, we hope to explore the robustness of neural operator approximate predictors for more challenging delay problems such as state-dependent and unknown noisy delays.  
\bibliography{references}
\bibliographystyle{IEEEtranS}

\renewcommand{\thesection}{Appendix \Alph{section}}
\setcounter{section}{0}
\section{Reverse Young's Inequality}\label{appendix:reverse-youngs}

\begin{lemma}
    Given $a, b \geq 0$ and $p, q$ such that 
    \begin{align}
        \frac{1}{p} - \frac{1}{q} = 1\,, 
    \end{align}
    then
    \begin{align}
        \frac{a^p}{p} - \frac{b^{-q}}{q} \leq ab\,. 
    \end{align}
\end{lemma}
\begin{proof}
    See \cite{proofwiki_reverse_young}. 
\end{proof}

\end{document}